\documentclass[conference]{IEEEtran}
\IEEEoverridecommandlockouts

\usepackage{cite}
\usepackage{amsmath,amssymb,amsfonts}
\usepackage{algorithmic}
\usepackage{graphicx}
\usepackage{textcomp}
\usepackage{xcolor}
\usepackage{multicol}
\def\BibTeX{{\rm B\kern-.05em{\sc i\kern-.025em b}\kern-.08em
    T\kern-.1667em\lower.7ex\hbox{E}\kern-.125emX}}

\usepackage{subfig}
\usepackage{booktabs}
\usepackage{hyperref}
\usepackage{float}
\usepackage[ruled, vlined, noend, boxed]{algorithm2e}
\usepackage{amsthm}
\usepackage{tcolorbox}

\newtheorem{theorem}{Theorem}
\newtheorem{lemma}{Lemma}[theorem]

\newtheorem{insight}{Finding}

\begin{document}

\title{\textsc{Mell}: Memory-Efficient Large Language Model Serving via Multi-GPU KV Cache Management
\thanks{}
}

\author{

\IEEEauthorblockN{Qianli Liu$^{1}$, Zicong Hong$^{1}$, Peng Li$^{2}$, Fahao Chen$^{3}$ and Song Guo$^{1}$}\IEEEauthorblockA{$^1$Department of Computer Science and Engineering, The Hong Kong University of Science and Technology, Hong Kong \\$^2$School
of Cyber Science and Engineering, Xi’an Jiaotong University, China\\$^3$School of Computer Science and Engineering, University of Aizu, Japan
\\qianli.liu@connect.ust.hk, ziconghong@gmail.com, pengli@xjtu.edu.cn, chenfh@ieee.org,
songguo@cse.ust.hk}

\thanks{

This research was supported by fundings from the Hong Kong RGC General Research Fund (152244/21E, 152169/22E, 152228/23E, 162161/24E), Research Impact Fund (No. R5011-23, No. R5060-19), Collaborative Research Fund (No. C1042-23GF), Theme-based Research Scheme (T43-518/24-N), National Natural Science Foundation of China (No. 62471383), Areas of Excellence Scheme (AoE/E-601/22-R), and Hong Kong Generative AI Research and Development Center from InnoHK. Corresponding authors: Zicong Hong, Song Guo. 

}
}

\maketitle

\begin{abstract}
Serving large language models (LLMs) for massive users is challenged by the significant memory footprint of the transient state, known as the \emph{key-value (KV) cache}, which scales with sequence length and number of requests. 
Instead of renting or buying more expensive GPUs, the load imbalance of the KV cache across GPUs, coupled with recent advances in inter-GPU communication, provides an opportunity to serve more requests via request migration. 
However, high migration overhead and unpredictable request patterns make it challenging. 
Therefore, this paper proposes \textsc{Mell}, a memory-efficient LLM serving system via \emph{multi-GPU KV cache management}. 
It saves the number of GPUs needed in the system by considering the dynamic KV cache load and the costly request migration. 
Specifically, we first develop an adaptive request migration mechanism to balance the computational and communication overheads and adapt to diverse resource conditions. 
Then, we design an online algorithm tailored to a multi-LLM request and multi-GPU scheduling problem with migration enabled.
It aims to minimise the required GPUs while limiting the number of migrations.
Finally, we implement a prototype of \textsc{Mell} and demonstrate that it reduces the number of GPUs by $31\%$ and increases the GPU utilization by $43\%$ at most compared to existing LLM serving systems.
\end{abstract}


\section{Introduction}
\label{sec:introduction}

The capability of Large Language Models (LLMs)~\cite{NEURIPS2020_1457c0d6,openai2024gpt4technicalreport,touvron2023llamaopenefficientfoundation} to understand and produce human-like text has established them as a central component of AI, dramatically improving many complex language-related tasks across industries. 
As the use of LLMs becomes more widespread, it is essential to deploy them on GPU clusters with a large number of GPUs~\cite{295545} and provide users with seamless access~\cite{orca,patel2024splitwise,10.1145/3627703.3629567,10229061,galaxy}.

To improve the LLM inference efficiency on GPUs, \emph{key-value (KV) cache} is one of the most critical modules~\cite{MLSYS2023_c4be71ab}.
It stores the keys and values of all previous tokens in GPU memory for each LLM inference to avoid redundant and repeated computations. 
Despite this advantage, there is a problem with GPU memory during long context processing and generation. Unlike the model weights, the KV cache is subject to size growth due to sequence length and batch size. 
As the demand for longer sequence lengths (along with larger batch sizes) grows~\cite{chen2024longlora}, the KV cache size problem becomes more pronounced.
Statistics show that the KV cache now often consumes over 30\% of the GPU memory~\cite{pagedattention}.

Memory management is important for accommodating more KV cache without renting or buying more expensive GPUs. Existing KV cache management works either compress the KV cache~\cite{keyformer,h2o,xiao2024efficient,liu2023scissorhands,liu2024kivi,ge2024model} or offload the KV cache to CPU memory~\cite{flexgen,InfiniGen,attentionstore,deepspeed_inference}. 
However, the former inevitably degrades LLM performance through quantization or sparsity, while the IO bottleneck between CPU memory and GPU limits the latter. 

To avoid these problems, we observe that the size of the KV caches on each GPU varies over time. 
Particularly, some GPUs are overwhelmed by the growth of the KV cache from running requests, while others have a lot of unused memory because the KV cache is released for completed requests.
This motivates us to schedule LLM requests with their KV cache from a heavily loaded GPU to a less loaded GPU to avoid renting or buying a new GPU. 
According to our preliminary experiments (see \textbf{Finding 3} in \autoref{sec:motivation}), such migration allows an LLM serving system to handle at most 60\% more LLM requests than that without migration.
A few works have demonstrated the feasibility of migrating requests across GPUs without significant service halt (i.e., live migration), e.g., Llumnix~\cite{llumnix} and ServerlessLLM~\cite{serverlessLLM}.

However, there are three challenges that need to be addressed regarding scheduling.
\textbf{1)~Unpredictable request patterns}: Besides the arrival time, the processing time of requests is difficult to learn due to the unpredictability of LLM response length.
Moreover, resource utilization changes as requests are processed due to updates in the KV cache.
\textbf{2)~High migration overhead}: Existing migration is either compute-intensive~\cite{serverlessLLM} or communication-intensive~\cite{llumnix} due to the KV cache transfer or re-prefill, respectively. Thus, the scheduling needs to balance the computational and communication overhead caused by migration. 
\textbf{3)~Theoretical guarantee}: most of the existing scheduling is based on a heuristic design (e.g., load swapping between GPUs with lowest load and highest load repeatedly~\cite{llumnix}) without a theoretical performance guarantee.

To solve these challenges, this paper proposes \textsc{Mell}, a memory-efficient LLM serving system enabled by a novel multi-GPU KV cache management. 
It perceives the system's dynamic KV cache load and resources, decides on the placement of LLM requests during their processing, and efficiently migrates the requests' KV cache to save the number of GPUs.

We summarize our contribution as follows.
\begin{itemize}
    \item We design an adaptive request migration mechanism by switching in real-time between the token and KV cache migration to balance computational and communication overheads for the dynamic environment.
    \item We develop an online KV cache scheduling algorithm in a multi-request multi-GPU environment to minimize the number of GPUs needed and limit the number of migrations.
    It has been rigorously proved to have a competitive ratio with the optimal solution of $4/3$ at most.
    \item We implement a prototype of \textsc{Mell} and demonstrate that it substantially reduces the number of GPUs by $9\% \sim 31\%$ and increases the GPU utilization by $10\% \sim 43\%$ compared to the existing LLM serving systems.
\end{itemize}
\section{Background \& Related Work} 

Modern LLMs, such as GPT~\cite{openai2024gpt4technicalreport} and LLaMA~\cite{touvron2023llamaopenefficientfoundation}, are based on the Transformer architecture and employ a decoder-only structure. 
\autoref{fig:kv_cache} shows a three-layer LLM, where nodes and edges indicate Transformer layers and dependencies between the layers, respectively. 
The Transformer layers are executed in the order denoted by the numbers, and the nodes that use the same set of model parameters (i.e., nodes representing the same layer) are filled with the same colour~\cite{TurboTransformers}.

\begin{figure}
    \centering
    \includegraphics[width=0.8\linewidth]{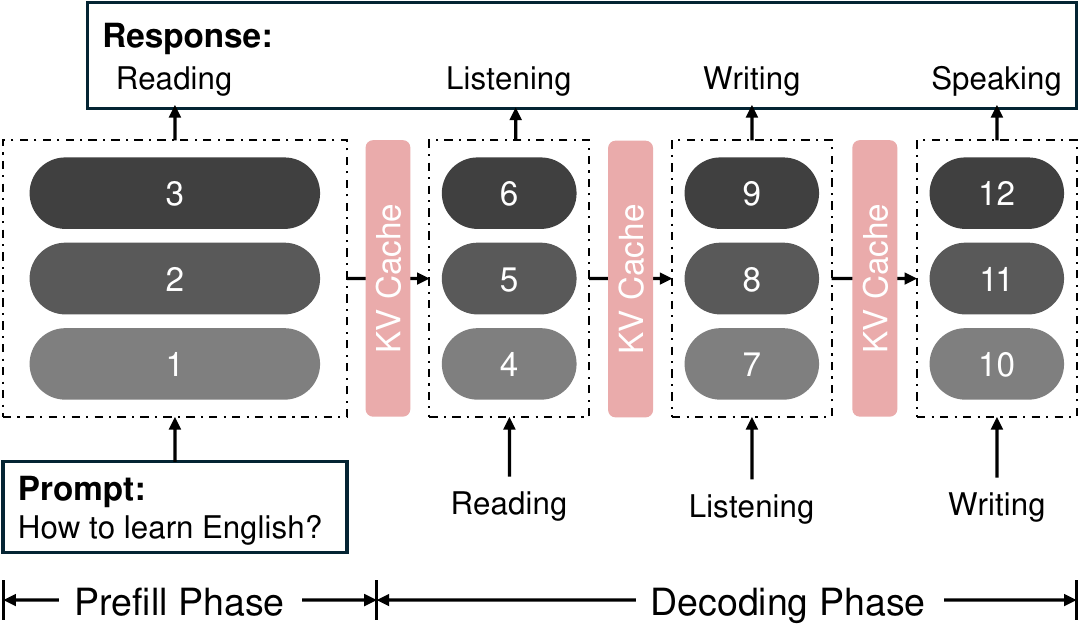}
    \caption{Serving procedure of an LLM request.}
    \label{fig:kv_cache}
\end{figure}

The processing of each LLM request is logically divided into a \emph{prefill} phase and a \emph{decoding} phase.
In the prefill phase, all input tokens designated as \emph{prompt} are processed in parallel. 
This phase generates the initial output token while storing the intermediate results of computed keys and values in the GPU memory, collectively referred to as the \emph{KV cache}. 
For example, in \autoref{fig:kv_cache}, a prompt ``How to learn English?'' generates the first token ``Reading'' and the KV cache.
The decoding phase then utilises this KV cache to generate new tokens autoregressively (i.e., ``Listening'', ``Writing'', and ``Speaking''), incorporating new keys and values into the KV cache. 

Despite avoiding recomputation, the KV cache exacerbates the huge memory consumption of LLMs. 
It has therefore been an active area of research in recent years, with numerous LLM serving systems proposed to address various aspects of KV cache management. 
\textbf{1) KV Cache Compression.} 
Substantial works save the memory consumption of the KV cache via quantization and sparsity~\cite{keyformer, ge2024model,h2o,xiao2024efficient,liu2023scissorhands,liu2024kivi}.
\textbf{2) Memory Management for KV Cache.}
To increase GPU utilization, several works propose efficient memory management for the KV cache~\cite{InfiniGen}.
Kwon \textit{et al.} propose PagedAttention~\cite{pagedattention} that allows KV cache to be stored in non-contiguous paged memory, reducing memory fragmentation.
Gao \textit{et al.} propose a hierarchical KV caching system that utilizes cost-effective storage media to store more KV caches~\cite{attentionstore}. 
\textbf{3) Request Migration across GPUs.}
To fully utilize the compute and memory resource, several works disaggregate each request's prefill and decode phase into separate GPUs~\cite{Taming,DistServe,patel2024splitwise,hu2024inferenceinterferencedisaggregatellm}.
However, they fix the placement of LLM requests during the memory-intensive decoding phase, even if there is a significant KV cache load imbalance.
Instead, Sun \textit{et al.} propose Llumnix~\cite{llumnix}, an LLM serving system that supports live migration for the KV cache across GPUs during the decoding phase. 
In other words, it introduces near-zero downtime by pipelining the computation and memory transfer.
Similarly, Fu \textit{et al.} propose ServerlessLLM~\cite{serverlessLLM}, a serverless LLM serving system via a two-stage live migration.
Moreover, Wu \textit{et al.} co-migrate requests and adapters for a LoRA LLM serving system~\cite{dlora}.

The first two types of work above focus on optimising the management of the KV cache within a single GPU, which is orthogonal to ours. 
Integrating these works can improve the memory efficiency of each GPU. 
The works most relevant to us are Llumnix~\cite{llumnix} and ServerlessLLM~\cite{serverlessLLM}, which support KV cache migration between GPUs. However, as discussed in \autoref{sec:introduction} and \autoref{sec:motivation}, several challenges need to be addressed, including unpredictable request patterns, high migration overhead, and theoretical guarantee. 
To overcome these challenges, our \textsc{Mell} develops a new multi-GPU KV cache management with an adaptive request migration mechanism for dynamic resource levels and an online KV cache scheduling algorithm that limits the number of GPUs and the number of migrations.

\section{Motivation}
\label{sec:motivation}

This section analyses the KV cache's characteristics, identifies its main bottleneck, and shows the potential for optimisation, which motivates the design of our \textsc{Mell}.

\begin{figure}[t]
    \centering
    \includegraphics[width=0.8\linewidth]{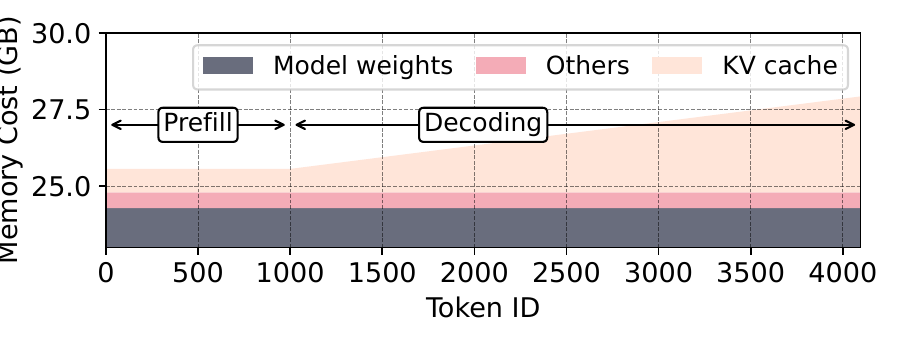}
    \caption{The memory cost of processing a request with $4096$ tokens on LLaMA-13B.}
    \label{fig:kvcache_usage}
\end{figure}

\begin{tcolorbox}
\begin{insight}
KV caches make existing LLM serving systems memory-bound, under-utilising GPU processing power and thus limiting serving throughput.
\end{insight}
\end{tcolorbox}

\begin{figure}[t]
    \centering
    \includegraphics[width=0.8\linewidth]{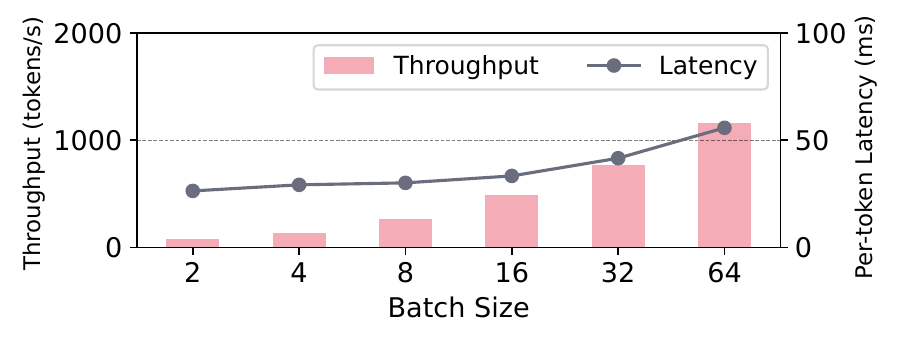}
    \caption{Throughput and per-token decoding latency of serving LLaMA-13B in a prompt length of $100$ as batch size increases.}
    \label{fig:decoding_batch}
\end{figure}

\autoref{fig:kvcache_usage} shows the memory cost of a request fed to the LLaMA-13B model~\cite{touvron2023llamaopenefficientfoundation} on an A100 GPU with $40$ GB of memory.
In the prefill phase, some memory space is pre-allocated for the KV cache according to the prompt length.
In the decoding phase, the size of the KV cache grows linearly with the increasing number of tokens generated.
For a request with a maximum length of $4096$ tokens, the memory cost of the KV cache is about $3.2$ GB.
After storing the model parameters of LLaMA-13B (roughly $24$ GB), the A100 GPU can only support a maximum batch of $5$ requests.
However, as shown in \autoref{fig:decoding_batch}, the latency per token remains relatively stable, and the throughput continues to increase as the batch size grows from $2$ to $16$ when memory is not bounded (in other words, the request length is short).

\begin{figure}[t]
	\centering
	\subfloat[][Vicuna-13B]{
		\begin{minipage}[t]{0.305\linewidth}
			\centering
			\includegraphics[width=\linewidth]{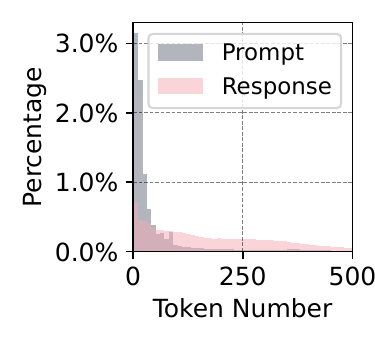}
		\end{minipage}
		\label{fig:input_output_0}
	}
	\hfill
	\subfloat[][Koala-13B]{
		\begin{minipage}[t]{0.305\linewidth}
			\centering
			\includegraphics[width=\linewidth]{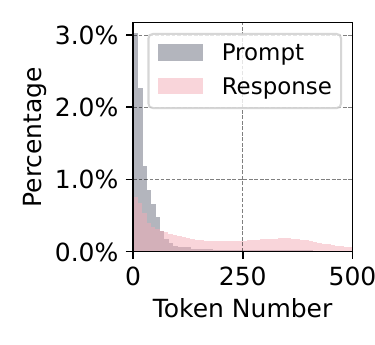}
		\end{minipage}%
		\label{fig:input_output_1}
	}
	\hfill
     \subfloat[][ChatGPT]{
    		\begin{minipage}[t]{0.305\linewidth}
    			\centering
    			\includegraphics[width=\linewidth]{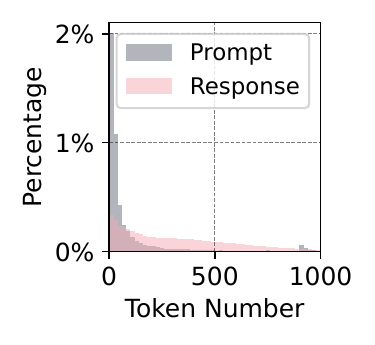}
    		\end{minipage}%
    		\label{fig:input_output_2}
    	}
	\centering
	\caption{The distribution of prompt and response length.}
    \label{fig:input_output}
\end{figure}

\begin{figure}[t]
    \centering
    \includegraphics[width=0.8\linewidth]{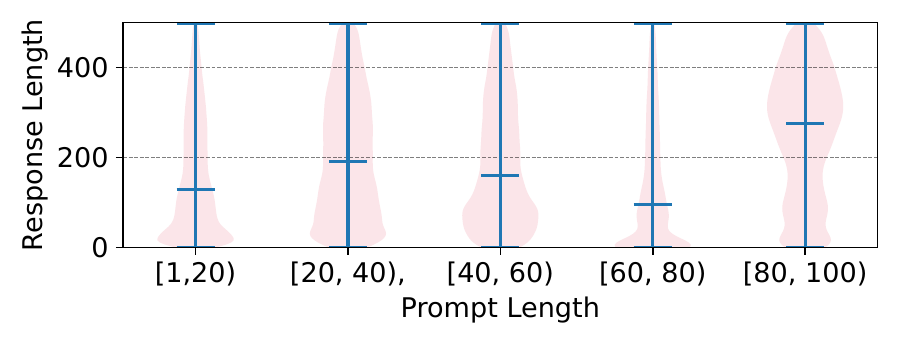}
    \caption{The distribution of response lengths under various prompt lengths in Vicunna-13B.}
    \label{fig:user2bot_0}
\end{figure}

\begin{figure}[t]
    \centering
    \includegraphics[width=0.75\linewidth]{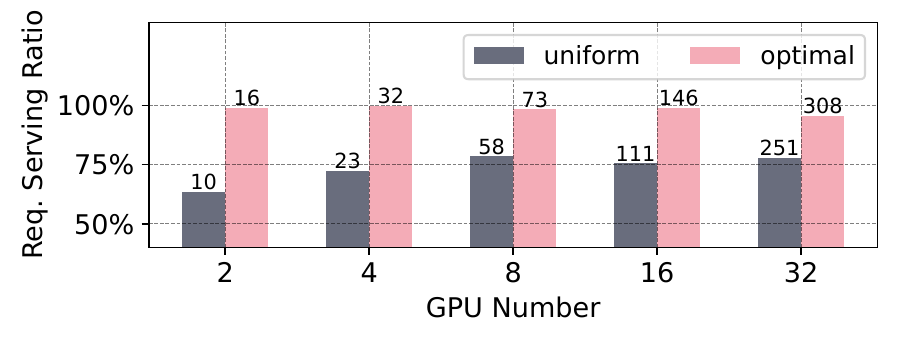}
    \caption{The request serving ratio of an LLM serving system without request migration and with migration enabled. The number above each bar denotes the number of requests served.}
    \label{fig:optimal_random}
\end{figure}

\begin{figure}[t]
	\centering
	\subfloat[][Llumnix]{
		\begin{minipage}[t]{0.43\linewidth}
			\centering
			\includegraphics[width=\linewidth]{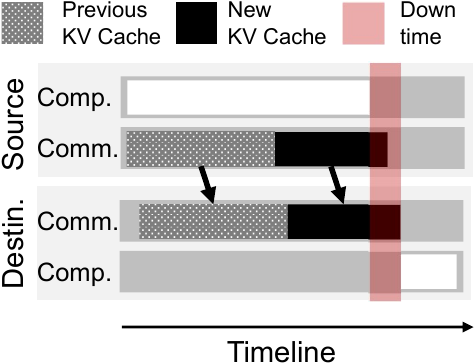}
		\end{minipage}
		\label{fig:llumnix}
	}
	\hfill
	\subfloat[][ServerlessLLM]{
		\begin{minipage}[t]{0.43\linewidth}
			\centering
			\includegraphics[width=\linewidth]{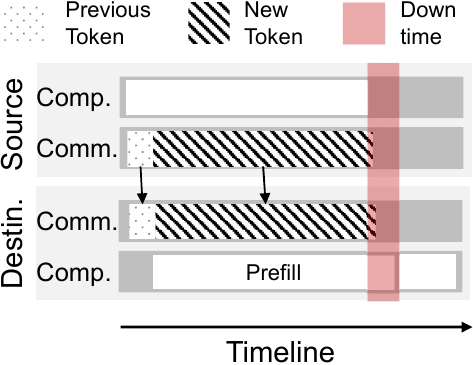}
		\end{minipage}%
		\label{fig:serverlessLLM}
	}
	\centering
	\caption{Two existing KV caches live migration approaches.}
    \label{fig:related_work}
\end{figure}

An intuitive idea to improve memory utilisation is to route each incoming request to a GPU with enough memory to hold the request's KV cache. However, this is difficult to implement, as discussed below.

\begin{tcolorbox}
\begin{insight}
It is difficult to predict the maximum KV cache required for an LLM request
due to the inherent unpredictability of the response length generated.
\end{insight}
\end{tcolorbox}

The preliminary experiment is conducted on LMSYS-Chat-1M~\cite{zheng2024lmsyschatm} and WildChat~\cite{zhao2024wildchat}, two large-scale datasets containing real chatbot conversations. 
\autoref{fig:input_output} shows the distribution of token numbers for responses in the three popular LLMs~\cite{zheng2023judging,koala_blogpost_2023,openai2024gpt4technicalreport}.
The distributions of token numbers vary considerably between different LLM models, with a wide range of possible values.
Moreover, as shown in \autoref{fig:user2bot_0}, the same prompt length can yield disparate response lengths. 
Consequently, it is challenging to ascertain the response length based on the length of the corresponding prompt.
Although a few existing works try to predict the response length~\cite{zheng2023response}, the prediction performance is poor (i.e., the accuracy is lower than 60\% even in a length range granularity of 100~\cite{hu2024inferenceinterferencedisaggregatellm}).

\textbf{Finding 2} highlights the challenge of scheduling algorithms without knowledge of KV cache size.

\begin{tcolorbox}
\begin{insight}
    Scheduling the placement of requests during their decoding phase can serve more requests simultaneously by fully utilizing multiple GPUs' memory for the KV cache. 
\end{insight}
\end{tcolorbox}

We evaluate the request serving ratio of an LLM serving system without request migration during the decoding phase and with migration enabled. The request serving ratio is the number of requests whose KV cache is retained in GPU memory. 
As shown in \autoref{fig:optimal_random}, request migration allows an LLM serving system to handle $23 \sim 60\%$ more requests than a system without request migration.

Several works propose live migration for KV caches to migrate the running LLM requests without long service interruption, such as Llumnix~\cite{llumnix} and ServerlessLLM~\cite{serverlessLLM}.
However, their scheduling is based on a heuristic design without a performance guarantee in theory. 
For example, Lluminx adopts a load balancing strategy between GPUs by swapping with the lowest load and highest load repeatedly~\cite{llumnix}.
Moreover, these works only focus on how to achieve the liveness of the migration but overlook the migration overhead on computation and communication resources as follows.

\begin{tcolorbox}
\begin{insight}
Existing migrations for the KV cache are either compute-intensive or communication-intensive.
\end{insight}
\end{tcolorbox}

\autoref{fig:related_work} shows the workflow of the LLM request migration in Llumnix~\cite{llumnix} and ServerlessLLM~\cite{serverlessLLM}. 
The live migration of Llumnix employs the intrinsic append-only nature of the KV cache to facilitate the concurrent transfer of the KV cache copy of the legacy tokens and the decoding computation for the new tokens.
However, each migration needs a huge amount of KV cache transferred between GPUs, burdening the inter-GPU bandwidth.
ServerlessLLM exhibits a comparable approach, albeit with a transformation between tokens and the KV cache, whereby the tokens are transmitted instead of the KV cache. 
Chunked prefill allows prefills to be batched together with decode requests.
However, the latency of requests in the decoding phase will slow down by up to 2.5x when they co-execute a migrating request with a prefill requirement~\cite{hu2024inferenceinterferencedisaggregatellm}.

These findings call for a new multi-GPU KV cache management in the LLM serving system, fully utilizing GPU memory and limiting the request migration number.

\section{System Overview}

\begin{figure}[t]
    \centering
    \includegraphics[width=0.75\linewidth]{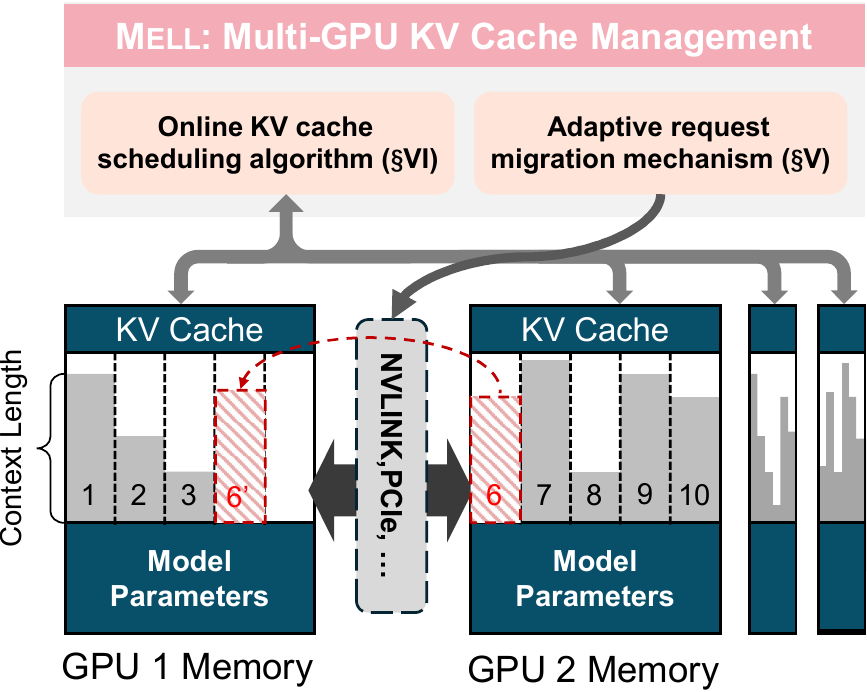}
    \caption{System architecture of \textsc{Mell}.}
    \label{fig:architecture}
\end{figure}

This section gives an overview of \textsc{Mell}'s design. First, we clarify three primary design goals of \textsc{Mell} as follows.
\begin{itemize}
    \item \textbf{GPU cost efficiency}: Given the high cost and scarcity of GPUs in the current market, the objective of \textsc{Mell} is to reduce the number of GPUs required to process LLM requests in a cost-effective manner.
    \item \textbf{Online strategy planning}: New requests may come at any time, and running requests may be completed. Thus, \textsc{Mell} needs to update the scheduling of the requests' KV caches with incomplete knowledge of the future.
    \item \textbf{Restricted-performance impact}: As a consequence of the request migration, the communication or computation resources are occupied and the normal running requests are influenced. Consequently, \textsc{Mell} mitigates the impact on the performance of the other requests.
\end{itemize}

A system overview of \textsc{Mell} is shown in \autoref{fig:architecture}, where \textsc{Mell} is integrated into the existing multi-GPU LLM serving framework.
\textsc{Mell} is not only a scheduling algorithm, but also a set of modules that optimize KV cache management across GPUs.
It employs two key components: an adaptive request migration mechanism (refer to \autoref{sec:migration}) and an online KV cache scheduler (refer to \autoref{sec:scheduling}). 
The former aims to balance the computational and communication resource overhead to minimise the negative impact caused by request migration.
The latter aims to minimise the number of GPUs by migrating the LLM cache of requests across GPUs with different workloads.

The system's life cycle is composed of multiple epochs.
At the beginning of each epoch, the instances send their state information (including request number and memory cost of each request's KV cache) to the cluster monitor of \textsc{Mell}.
According to the state information, \textsc{Mell} generates the updated KV cache scheduling strategy via the online scheduling algorithm and sends the strategy to the instances.
Then, according to the given strategy and the communication and computation capacity of the system, the instances can migrate the requests by following the adaptive request migration mechanism. 
For example, in \autoref{fig:architecture}, GPU 2 migrates request 6 to GPU 1 to reduce the load on GPU 2.

\section{Adaptive Request Migration}
\label{sec:migration}

\begin{figure}[t]
    \centering
    \includegraphics[width=0.75\linewidth]{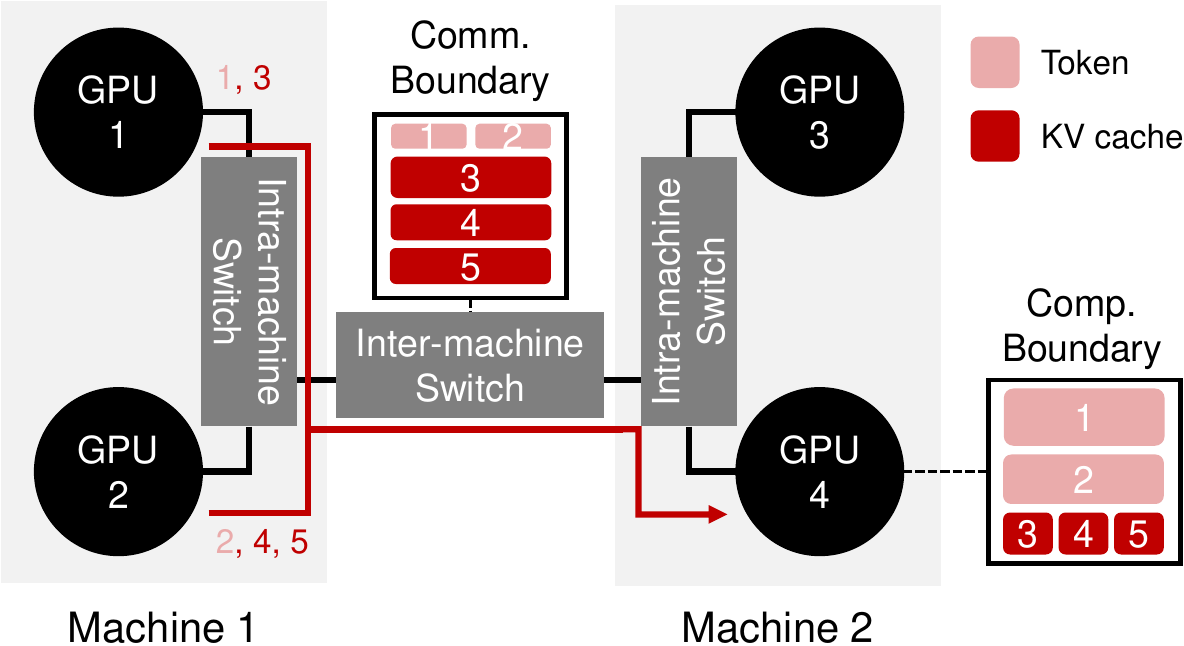}
    \caption{An example of adaptive request migration.}
    \label{fig:hybrid_migration}
\end{figure}

As discussed in \textbf{Finding 4} in \autoref{sec:motivation}, existing LLM request migrations are either compute-intensive (i.e., token transfer) or communication-intensive (i.e., KV cache transfer). This fixed migration strategy can cause resource congestion when multiple requests must be migrated simultaneously. Therefore, mitigating the negative impact of request migration by balancing its communication and computation costs becomes a critical challenge that needs to be addressed by \textsc{Mell}.

To address this challenge, we propose an adaptive request migration mechanism in \textsc{Mell}. Its main idea is to first identify the idle computational and communication resources in the system that can be used by the request migration without affecting the normal operation of the system. Next, it migrates each request by transferring either tokens or KV cache, and orchestrates all requests to be migrated to make the consumption within the boundary. As shown in \autoref{fig:hybrid_migration}, its workflow consists of the following steps.

\textbf{Boundary Profiling.} We first define the \emph{communication boundary} of a GPU communication link as a certain amount of data that can be transferred over the link. 
This boundary is set to ensure that the data can be transferred in a limited amount of time. 
This is because the GPU memory occupied by the migrating requests cannot be released until the transfer is complete.
We also define the \emph{computation boundary} of an instance with a batch size as a certain number of tokens to be prefilled due to migration.
This is because a prefill computation with long tokens interferes with the co-located computation, while one with short tokens does not since it can use idle computation resource~\cite{hu2024inferenceinterferencedisaggregatellm}.
At the beginning of the system, we identify the communication boundary for every link and the computation boundary for every instance via offline profiling.
The boundary information will be shared with all instances in the system.

\textbf{Hybrid Migration.} Given a new KV cache scheduling strategy, each instance migrates requests according to the boundary. 
Each instance needs to divide its requests to be migrated into two classes. 
The first class includes the requests transferred as KV cache~\cite{llumnix}. 
The second class includes the requests transferred as tokens and then prefilled in the destination instances~\cite{serverlessLLM}.
The division is formulated as a two-bin-packing problem that can be solved using a greedy algorithm (e.g., first-fit or best-fit).

\textbf{Global Consensus.} Multiple instances can use the same link and migrate requests to the same instance.
For example, in \autoref{fig:hybrid_migration}, GPU 1 and GPU 2 use the inter-machine switch to migrate requests to GPU 4. 
They may exceed the boundaries if the above division is done without cooperation.
To avoid this, each instance runs the algorithm considering all requests to be migrated in the system, not just its own requests.
The global division is still a two-bin packing problem, where each request can choose either token transfer or KV cache transfer.

\section{Online KV Cache Scheduling}
\label{sec:scheduling}

This section presents the online KV cache scheduling algorithm in \textsc{Mell}. 
The system model of a multi-GPU LLM serving system is first presented, followed by details of our algorithm and theoretical performance analysis.


\subsection{System Model}
\label{system-model}

\paragraph{GPU Cluster}
We consider a set of homogeneous GPUs, denoted by $J$ and the memory capacity for the KV cache in each GPU is denoted by $C$. 

\paragraph{Requests} Users can send a set of LLM requests denoted by $I$ to the system within time slots $T$, and the arrival time of request $i \in I$ is denoted by $a_i \in T$. 
These requests have different token numbers to be processed (i.e., prefilled and decoded) due to the diverse tasks of users.
The memory usage of the KV cache of request $i \in I$ at time $t \in T$ is $S_i^t$, and it linearly increases with the number of tokens processed before the request is completed, i.e., $S_i^t \ge S_i^{t-1}$.

\paragraph{Serving Strategy} 
When request $i$ arrives, it is assigned to a GPU for processing. 
We denote a serving strategy for $T$ as $\mathbf{x} = \{\mathbf{x}^t\}_{t \in T}$ in which $\mathbf{x}^t = \{x_{i,j}^t\}_{i \in I, j \in J}$ is defined as
\begin{equation}
	x_{i,j}^t=
	\begin{cases}
	1, &\text{if request $i$ runs on GPU $j$ at time $t$}\\
	0, &\text{otherwise}
	\end{cases}
\end{equation}

Given a strategy $\mathbf{x}$ for $T$, the system should ensure that any GPU at any time slot must have enough memory space to process the allocated request, i.e.,
\begin{equation}
    \sum_{i \in I} x_{i,j}^t S_i^t \leq C, \quad \forall j \in J, \forall t \in T.
\label{eq:gpu_memory}
\end{equation}

\paragraph{System Cost} $y_j^t$ indicates whether any requests are running on GPU $j$ at time $t$ and $y_j^t = \min \{1, \sum_{i \in I} x_{i,j}^t\}$.
Therefore, the number of GPUs needed to serve the set of inference requests $I$ by strategy $\mathbf{x}$ within time slots $T$ is 
\begin{equation}
    B(\mathbf{x})=\max_{t \in T} \sum_{j \in J} y_j^t.
\end{equation}



\subsection{Problem Formulation}

For a serving strategy $\mathbf{x}$ to minimize the number of GPUs needed for a set of requests $I$, we formulate a \emph{KV cache scheduling problem} in a multi-node multi-request system
\begin{subequations}
    \begin{align}
         \min B(\mathbf{x}) = &\min \max_{t \in T} \sum_{j \in J} y_j^t \tag{6}
        \label{eq:p1_obj}\\
         {\rm s.t.} \sum_{i \in I}& x_{i,j}^t S_i^t \leq C, \quad \forall j \in J, \forall t \in T.\notag
    \end{align}
    \label{problem:p1}
\end{subequations}


Solving the problem poses the following challenges.  
First, the arrival and completion of requests are unpredictable, limiting decision-making to the available information. This makes it difficult to approach the global optimum, as the system cannot anticipate all future needs. 
Second, the memory required for each request $S_i^t$ is dynamic, growing over time until the request is completed. This requires constant adjustment of resource allocations, making it difficult to plan and optimise resource usage.
Third, the decision at any moment resembles a bin packing problem, which is NP-complete~\cite{NPC}. Our problem, however, introduces greater complexity as historical choices influence each decision, complicating the resolution process significantly beyond the NP-complete framework.

\begin{figure*}[t]
    \centering
    \noindent\framebox{
    \begin{minipage}{0.99\linewidth}
    \begin{multicols}{2}
    \normalsize
    $\bigstar J.Allocate(i)$: \\
    \small
    \vspace{-0.1in}
    
    \begin{algorithmic}[1]
    \STATE \textit{\textbf{Allocate T-request:}} 
        For all L-GPU $j\in J$ with enough memory to fit $i$, allocate $i$ to GPU $j$ with the highest priority. 
        Otherwise, allocate $i$ to the most recently activated T-GPU.
    
    \STATE \textbf{\textit{Allocate S/M-request:}}
        For all L-GPU $j\in J$ with $S_{i_L}^t+S_i^t < C$, $i_L$ is the L-request in $j$. 
        Allocate $i$ to $j$ with the highest priority.
        Depart and re-allocate any T-request that exists in $j$.
        Otherwise, allocate $i$ to the most recently activated S/M-GPU.
    
    \STATE \textbf{\textit{Allocate L-request:}}
        Activate a new GPU $j$, $J = J \cup \{j\}$. 
        Allocate $i$ to $j$.
        Move an S/M-request from an S/M-GPU $j'$ to $j$ if possible. Then fulfil $j'$ with S/M-request from the most recently activated S/M-GPU. 
    
    \end{algorithmic}
    \BlankLine
    \normalsize
    $\bigstar J.Depart(i)$: \\
    \small
    $~~~~~$ Assume request $i$ is processed by GPU $j\in J$ currently.
    
    \begin{algorithmic}[1]
    \STATE \textit{\textbf{GPU $j$ is the most recent activated GPU:}} Remove $i$ from $j$. 
    \STATE \textit{\textbf{Depart T-request:}}
        If $j$ is T-GPU, move a T-request from the most recently activated T/M-GPU to $j$. 
        Otherwise, move a T-request from the most recently activated T-GPU to fit in $j$.
    \STATE \textit{\textbf{Depart S/M-request:}}
        If $j$ is an S/M-GPU, move an S/M-request from the most recently activated S/M-GPU to $j$. Re-allocate any T-request that may exist in $j$.
        If $j$ is an L-GPU, move an S/M-request to $j$ from an S/M-GPU $j'$ with the highest priority for GPU $j$. Then, fulfil $j'$ with S/M-request from the most recently activated S/M-GPU. 
    \STATE \textit{\textbf{Depart L-request:}}
        Depart and re-allocate all other requests in $j$.
    \end{algorithmic}
    
    \BlankLine
    \normalsize
    $\bigstar J.Update(i)$: \\
    \small
    $~~~~~$ Assume request $i$ is processed by GPU $j\in J$ currently.
    
    \begin{algorithmic}[1]
    \STATE \textbf{\textit{T/S-request $\rightarrow$ S/M-request:}}
    Depart $i$ and re-allocate $i$.
    
    \STATE \textbf{\textit{M-request $\rightarrow$ L-request:}}
    If $j$ is a L-GPU, depart $i$ and re-allocate $i$.
    If $j$ is an M-GPU and overload occurs after the update, depart and re-allocate all other requests in $j$.
    
    \STATE \textbf{\textit{L-request $\rightarrow$ L-request:}}
    If overload occurs after growth, depart and re-allocate all other requests in $j$.
    \end{algorithmic}
    
    \end{multicols}
    \vspace{0.0001in}
    \end{minipage}}
    \vspace{-.05in}
    \caption{Three request operations for request $i$ on GPU cluster $J$ for multi-GPU KV cache scheduling.}
    \label{fig:operations}
    \vspace{-.2in}
\end{figure*}

\subsection{Online Algorithm Design}

According to the characteristics of LLM serving, we design an online algorithm for the KV cache scheduling problem motivated by~\cite{adaptive_online,efficient_online,10.1145/2528521.1508269}. 
It allocates incoming LLM requests based on GPU memory and request requirements, then updates allocations by migrating requests between GPUs to adapt to workload fluctuations. 
Unlike existing scheduling algorithms focusing primarily on immediate state changes, our algorithm takes a long-term view of scheduling to minimize space fragmentation and avoid creating unused fragmented spaces. 
It can achieve near-optimal allocation globally rather than just providing short-term solutions.

\textbf{Priority-aware GPU Categories.}
We first classify the requests into four categories based on their KV cache sizes: $L$ (Large), $M$ (Medium), $S$ (Small), and $T$ (Tiny). 
Request $i$ is an $L$-request if $S_i^t$ is within $(C/2, C]$; $M$-request if $S_i^t$ is within $(C/3, C/2]$ ; $S$-request if $S_i^t$ is between $(C/4, C/3]$; 
$T$-request if $S_i^t$ ranges from $(C/8, C/4]$. 
For requests smaller than $C/8$, we group them into multi-items with sizes in the range $(C/8, C/4]$.
GPUs are categorised based on the largest type of request they process: GPU $j$ is labelled as an L, M, S, or T-GPU if its largest request in category $j$ is an L, M, S, or T-request.
We also define a priority relationship from GPUs $j$ to $j'$, which is determined by factors including the workload (e.g., request number and idle GPU memory) of GPU $j'$ and the distance between GPUs $j$ to $j'$.
For example, a GPU $j'$ that handles fewer requests, has more GPU memory, or is on the same machine as GPU $j$ will be assigned a higher priority for GPU $j$. 
The weights of different factors are set by the LLM service provider.
Besides, we define a priority of GPU $j$ that is only determined by the workload of GPU $j$ for the allocation of incoming requests.

\begin{algorithm}[t]
\SetAlgoLined 
\caption{Overall Workflow}
\label{alg:online}
\KwIn{LLM request set $I$, GPU cluster $J$}
\For{$t \in T$}{
    \For{$i \in \{ i \mid S^t_i > 0 \lor S^{t-1}_i > 0 \}$}{
    \If{\textnormal{request $i$ arrives at $t$}}{
         $J.Allocate(i)$
    }
    \ElseIf{\textnormal{request $i$ is completed at $t$}}{
         $J.Depart(i)$
    }
    \ElseIf{\textnormal{request $i$'s type changes at $t$}}{
         $J.Update(i)$
    }
    }
    \For{GPU $j\in J$ processing no request}{
        terminate GPU $j$, $J = J-\{j\}$
    }
}
\end{algorithm}

\textbf{Request Allocation/Depart/Update.}
The arrival of new requests, the departure of completed requests, and the growth of the KV cache of running requests can all lead to GPU underloaded or overloaded. 
Therefore, \autoref{alg:online} takes the LLM request set $I$ and GPU cluster $J$ as input to update the allocation based on three operations: (1) allocating new requests to the most appropriate GPU, (2) dropping completed requests to free resources, and (3) updating the position of running requests along with their processing. 
Details of each operation are given in \autoref{fig:operations}, which guarantees that our algorithm maintains a near-optimal number of GPUs, as proven in \autoref{analysis}.


\textbf{Request Operation Batching.}
The request operations in \autoref{fig:operations} are designed to efficiently manage individual requests' allocation, departure, and update. 
However, overlapping operations can occur when multiple requests need to perform these operations simultaneously, resulting in redundant request migration. 
To address this problem, we introduce \textit{request operation batching}. 
This approach combines and optimizes operations as a unified group rather than discretely, minimizing unnecessary resource allocation and migration. 
Implementing operation batching is critical to ensure efficient request migration within our framework, especially in high-demand scenarios.
Given an operation set $O$, the steps for operation batching are: 
(1) Execute all $Depart()$ in $O$. Instead of executing the possible migration caused by $Depart()$, add them into an operation buffer $B$. 
(2) Execute all $Update()$ in $O$. Instead of executing the possible migration caused by $Update()$, add them into buffer $B$. Check $B$ and remove unnecessary movement.
(3) Execute all $Allocate()$ in $O$. Check $B$ and remove unnecessary movement.
(4) Execute all operations in the buffer.

\section{Analysis}
\label{analysis}

\begin{theorem}
\label{theorem1}
    Given any set of LLM requests $I$, the allocation obtained by our algorithm satisfies all the following properties (with a constant number of exceptions):
    \begin{enumerate}
        \item M-GPU process two M-requests, possibly one T-request.
        \item S-GPU process three S-requests.
        \item T-GPU memory usage is at least 75\%.
        \item L-GPU $j$ process no S/M-request only if no M/S-request in the M/S-GPU can fit in $j$.
        \item T-GPU exist only if all GPU memory utilisation of L/M-GPU is at least 75\%.
    \end{enumerate}
\end{theorem}
\begin{proof}
In the following, we discuss every operation separately. 
\textit{1) Allocate/Depart T-request:} The preference for the L-GPU in $Allocate()$ and the attempt to replenish the L-GPU with M/S-requests in $Depart()$ ensure that property 4 is not violated.
The remaining properties remain unaffected.
\textit{2) Allocate/Depart M/S-request:} The preference for L-GPU in $Allocate()$ and the attempt to refill the L-GPU with M/S-request in $Depart()$ ensure that property 4 is not violated. 
The remaining properties remain unaffected.
\textit{3) Allocate/Depart L-request:} Allocation/departure of L-request triggers $Allocate()/Depart()$ of other types of requests, which are shown to satisfy the properties. 
\textit{4) Update Operation:} $Update()$ consist of $Allocate()$ and $Depart()$; hence, there is no violation of properties.
\end{proof}

\begin{theorem}
\label{competitive-ratio}
    For any request set $I$, given the scheduling algorithm $A$ that maintains the allocation fulfils all the properties in \autoref{theorem1}, the competitive ratio of $A$ is at most $4/3$.
\end{theorem}

Intermediate value \textit{weight} will be introduced for the following proofs. 
Each request $x$ will have a corresponding weight $w(x)$, the total weight $W(I)$ of LLM request set $I$ is the sum of the weights of all the requests, i.e. $W(I) = \sum_{x\in I} w(x)$. 
In addition, we divide the L-requests into two types: \textit{Single} type if there is no M/S-request in this L-GPU, otherwise \textit{Combined}.  
The number of \textit{single} request is $\mathcal{S}$ and the number of \textit{combined} request is $\mathcal{C}$.
The weight of a single L-request is $w(x) = 1$; combined L-request is $w(x) = 5/6$; M-request is $w(x) = 1/2$; S-request is $w(x) = 1/3$; T-request is $w(x) = 0$, i.e. the T-requests lead to no difference in weight.

\begin{lemma}
\label{lemma1}
    Given a scheduling algorithm $A$ and request set $I$, the allocation is denoted as $A(I)$. If $A(I)$ fulfils all the properties in \autoref{theorem1}, we have $|A(I)| \le W(I)+c$, where $|A(I)|$ is the number of GPUs in $A(I)$ and $c$ is a constant.
\end{lemma}
\begin{proof}
    To prove $|A(I)| \le W(I)+c$, we need to prove the average weight of GPUs in $A(I)$ is greater or equal to 1.
    By the properties in \autoref{theorem1}, we know:
    \begin{itemize}
        \item Weight of M-GPU is $1/2+1/2 = 1$.
        \item Weight of S-GPU is $1/3+1/3+1/3 = 1$.
        \item Weight of an L-GPU containing a single L-request is $1$.
    \end{itemize}
    
    The following will prove that the average weight of L-GPUs handling M/S requests is greater than or equal to 1. By the definition of the couple L-requests, it is easy to see that at least $\lfloor \frac{\mathcal{C}}{2}\rfloor$ combined L-requests can fit with $\lfloor \frac{\mathcal{C}}{2}\rfloor$ M/S-requests. So the total weights of L-GPUs handling M/S requests are at least $\frac{\mathcal{C}}{2}*1/3+\frac{\mathcal{C}}{2}*5/6 = \mathcal{C}$. So $A(I)$ is bounded by $W(I)+c$.
\end{proof}

\begin{lemma}
\label{lemma2}
    Given request set $I$, $OPT(I)\le 3/4 W(I)$ where $OPT(I)$ is the optimal allocation of $I$.  
\end{lemma}
\begin{proof}
    \begin{table}[H]
        \centering
        \begin{tabular}{cc}
            \toprule
            Possible Combination & \makebox[3cm][c]{Weight}  \\
            \midrule
            L,LT,LTT & $1<4/3$ \\
            LM & $5/6+1/2 = 4/3$\\
            LS & $5/6+1/3 < 4/3$\\
            MM & $1/2+1/2 <4/3$\\
            MSS & $1/2+1/3+1/3 < 4/3$\\
            SSS & $1/3+1/3+1/3 < 4/3$\\
            \bottomrule    
        \end{tabular}
        \caption{Weight of GPUs in each type}
        \label{weight}
    \end{table}
    To prove $OPT(I)\le 3/4 W(I)$, the weight of all the possible GPUs in $OPT(I)$ is shown in \autoref{weight}
\end{proof}
\begin{proof}[Proof of Theorem \ref{competitive-ratio}]
    The allocation for the request list $I$ generated by algorithm $A$ is denoted as $A(I)$. The total size of requests $i\in I$ is $S(I) = \sum_{i\in I} S_i$. The proof will be divided into the following two separate cases:\\
    \textit{Case 1: There is T-GPU in $A(I)$.} 
    By the fifth property, all the L-GPU and M-GPU are at least $3/4$ full. And the S-GPU is also at least $3/4$ full since the size of S-requests is in $(C/4, C/3]$ and each S-GPU processes 3 S-requests. Thus all the GPUs are at least $3/4$ full (except a constant number of \textit{latest GPUs}). Therefore we have:
    $|P_A(I)|\le\frac{4}{3} S(I) \le \frac{4}{3} \textsc{OPT}(I)$.

            
    \textit{Case 2: There is no T-GPU in $A(I)$}.
    From ~\autoref{lemma1} and ~\autoref{lemma2}, we can conclude that the inequality $|A(I)| \le W(I)+c \le 4/3\cdot OPT(I)+c$ holds.
\end{proof}

\begin{theorem}
\label{theorem3}
    Given a GPU allocation $A(I)$ of requests set $I$, which fulfils all the properties in \autoref{lemma1}, the maximum number of request migrations caused by an operation is ten.
\end{theorem}
\begin{proof}
    We discuss the operations separately. Noted that allocation and departure in the latest GPU do not result in any request migration; therefore, in the following discussion, the scope of the non-latest GPU is discussed by default.
    \begin{itemize}
        \item Allocate/Depart T-request (2 migrations):
        It is easy to observe that no request migration will be caused by allocating T-request. 
        When departing a T-request from GPU $j$. 2 migration may caused by the departure: use another T-request $i_T$ to replenish GPU $j$ and use another T-request to replenish the T-GPU processing $i_T$.
        So allocate/depart T-request will lead to at most 2 migrations.
        \item Allocate/Depart M/S-request (5 migrations): 
        Allocate M/S-request to L-GPU may cause up to two T-requests to be removed and re-allocated (two T-requests in L-GPU).
        Departing an S-request from an S-GPU may cause one migration of another S-request from the latest S-GPU.
        Departing an M-request from M-GPU $j$ may cause 1 migration of another M-request from the latest M-GPU to $j$ and trigger a departure of T-request, i.e., in total 4 migrations.
        Departing M/S-request from L-GPU will need to find another M/S to fit in the L-GPU. Worst case will trigger a departure of M-request, which may cause at most 4 migrations, i.e., in total 5 migrations.
        \item Allocate/Depart L-request (5 migrations):
        The worst case of allocating an L-request triggers a departure of an M-request and migrates it to the new GPU, i.e. may cause up to 5 migrations. Departing an L-request may cause two T-requests or one M/S-request to be departed and reallocated. Therefore, it may cause up to 3 migrations.
        \item Update (10 migrations): The update operation will trigger a departure of the original type request and an allocation of the new type request, so the worst case is departing L-request (5 migrations) and allocating L-request (5 migrations), i.e., at most 10 migrations.
    \end{itemize}\end{proof}

\section{Experiment}
\label{sec:experiment}

\subsection{Implementation}

We implement a prototype of \textsc{Mell} on top of vLLM~\cite{vllm}, representing the state of the art in serving systems and offering some advanced features, including paged attention and continuous batching.
We deploy each instance by Ray~\cite{ray} actor to implement GPU workers that execute the LLM inference and schedule the instance.
The request migration of \textsc{Mell} is supported by the point-to-point GPU communication of tokens and KV cache in Gloo~\cite{gloo} similar to Llumnix~\cite{llumnix}.

\subsection{Experimental Setup}

\textbf{Workloads.} We evaluate \textsc{Mell} based on LLaMA2 7B and 13B~\cite{touvron2023llamaopenefficientfoundation}, one of the most popular open-sourced LLMs. 
The LLM request inputs are based on LMSYS-Chat-1M~\cite{zheng2024lmsyschatm} and WildChat~\cite{zhao2024wildchat}, two datasets containing more than one million real conversations collected from chatbot applications. 
To simulate the state-of-the-art LLMs with long-context (e.g., GPT-4o~\cite{openai_model} and Claude 3.5 Sonnet~\cite{claude_model}), we scale up each conversation by a factor of ten.
The LLM request arrival pattern of the workload is generated from the following data. 
First, to simulate frequent, middle, and infrequent workloads, we use three Poisson distributions with a setting of $\lambda = 0.5$, $0.8$, and $1.1$, respectively. 
Second, we use the traces from multiple LLM inference services in Azure~\cite{patel2024splitwise} collected on November 11th 2023, to simulate production-like workload arrival patterns and characteristics.

\textbf{Node setup.}
We deploy \textsc{Mell} on a small-scale GPU cluster, including 8 NVIDIA GeForce RTX 4090 GPUs, each with 24 GB of memory, and 4 NVIDIA A100 GPUs, each with 40 GB of memory. 
The intra-machine GPU communication is PCIe 4.0, while the inter-machine GPU communication is 10 Gbps.
We use this testbed to collect traces of request processing speeds and inter-GPU bandwidth under different workloads.
We use this testbed to collect traces of request processing speeds and inter-GPU bandwidth under various workloads. 
We then use them to simulate the deployment of \textsc{Mell} on a large-scale cluster for evaluation.

\begin{figure}[t]
	\centering
	\subfloat[][LLaMA-13B on NVIDIA A100]{
		\begin{minipage}[t]{\linewidth}
			\centering
			\includegraphics[width=0.95\linewidth]{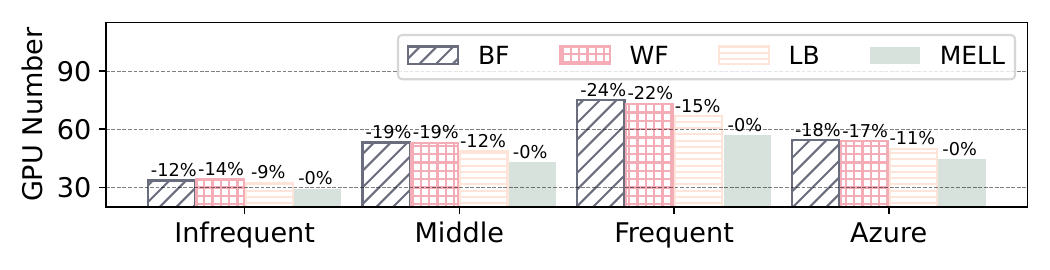}
		\end{minipage}
		\label{fig:GPU_num_A100_13B}
    	}
	\\
	\subfloat[][LLaMA-7B on NVIDIA 4090]{
		\begin{minipage}[t]{\linewidth}
			\centering
			\includegraphics[width=0.95\linewidth]{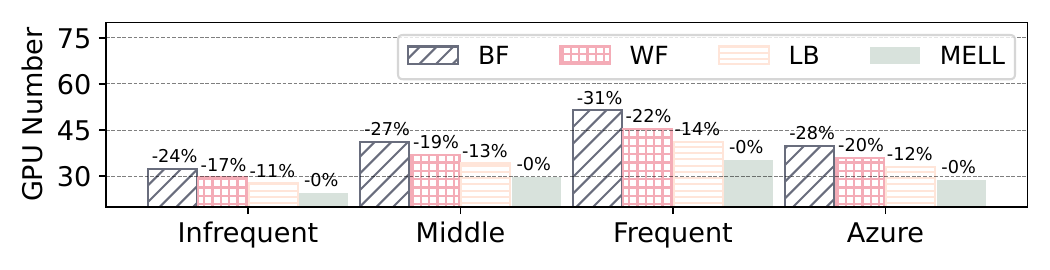}
		\end{minipage}%
		\label{fig:GPU_num_4090_7B}
	}
	\centering
	\caption{The number of GPUs needed by different systems under the Poisson and Azure workloads. The number above the bar denotes the difference between the baseline and \textsc{Mell}. 
    }
    \label{fig:GPU_num}
\end{figure}

\textbf{Baseline.} To evaluate the efficiency of \textsc{Mell}, we have conducted a comparative analysis with the following algorithms.
(1) BF: A scheduling algorithm dispatches each incoming request to the GPU with the least but sufficient memory (i.e., \underline{B}est-\underline{F}it) and does not migrate running requests between GPUs.
(2) WF: A scheduling algorithm dispatches each incoming request to the least memory (i.e., \underline{W}orst-\underline{F}it) and does not migrate running requests between GPUs.
A similar algorithm is widely adopted by existing LLM serving systems~\cite{DistServe,patel2024splitwise}.
(3) LB: A scheduling algorithm dispatches incoming requests to the GPU with the least memory (i.e., worst fit) and achieves \underline{L}oad-\underline{B}alancing via request migration by transferring the KV cache between GPUs, adopted by LLumnix~\cite{llumnix}.
These algorithms activate a new GPU if no GPU can handle an incoming request and terminate a GPU if it is idle. We implement these algorithms in our system to ensure a targeted comparison for scheduling requests.

\textbf{Metrics.} We evaluate \textsc{Mell} and the baselines based on the following metrics. (1) The number of GPUs required by the LLM service provider to serve user requests. (2) The migration frequency (i.e., migrations per second) required by the LLM serving system to fully utilize idle GPU memory. (3) The GPU utilization (i.e., percentage of GPU memory in use).

\subsection{Results}

\begin{figure}[t]
	\centering
	\subfloat[][LLaMA-13B on NVIDIA A100]{
		\begin{minipage}[t]{\linewidth}
			\centering
			\includegraphics[width=0.95\linewidth]{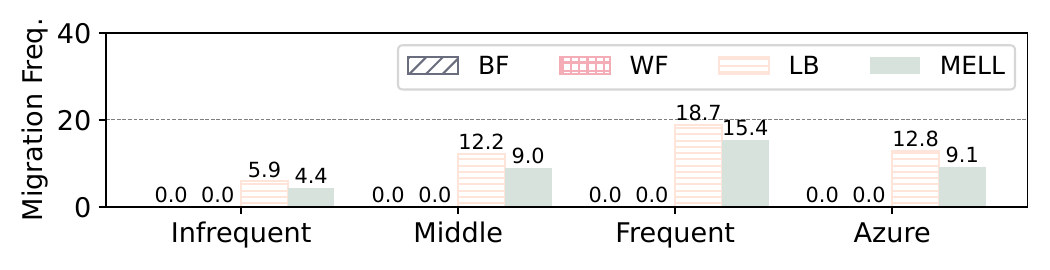}
		\end{minipage}
		\label{fig:A100_13B_move_num}
	}
	\\
	\subfloat[][LLaMA-7B on NVIDIA 4090]{
		\begin{minipage}[t]{\linewidth}
			\centering
			\includegraphics[width=0.95\linewidth]{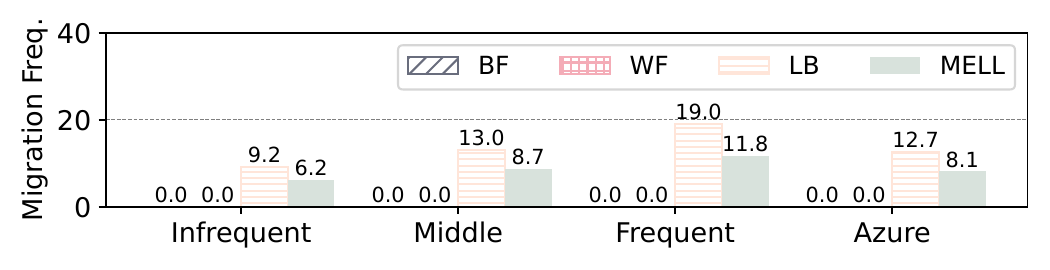} 
		\end{minipage}%
		\label{fig:4090_7B_move_num}
	}
	\centering
	\caption{The migration frequency in different systems under the Poisson and Azure workloads. }
    \label{fig:move_num}
\end{figure}

We evaluate the number of GPUs required by various LLM serving systems under different workloads, as shown in \autoref{fig:GPU_num}. 
While the BF and WF algorithms perform similarly, the LB algorithm outperforms both because it supports the request migration, balances GPU load, and improves GPU utilization. 
\textsc{Mell} combines the advantages of request migration with a design that reduces GPU fragmentation, reducing GPU demand by up to 15\% compared to LB and over 20\% compared to BF and WF. 
\textsc{Mell} increases its advantage as the workload's frequency grows. It is because, as the requirement of GPU grows, the existing algorithms result in more fragmented space, providing \textsc{Mell} greater room for optimization.
Moreover, this improvement is particularly evident when using the 4090 GPU for LLaMA-7B, where the limited KV cache storage causes larger fluctuations in GPU demand, emphasizing the need for an online KV cache scheduling algorithm of \textsc{Mell}.

We evaluate the migration frequency in various LLM serving systems under different workloads.
As illustrated in \autoref{fig:move_num}, \textsc{Mell} consistently exhibits a lower migration frequency than LB, because \textsc{Mell} is designed with an upper limit on the number of migrations according to \autoref{theorem3}. The long-term consideration inherent to the scheduling process represents a significant advantage of \textsc{Mell} over LB. 
Additionally, the proposed operation batching can effectively reduce the incidence of unnecessary migrations.
As seen in \autoref{fig:batch}, the technique reduces the number of migrations by up to $30\%$ under the Poisson workloads and $25\%$ under the Azure workload.
Also, only LB and \textsc{Mell} support migration of running requests; BF and WF do not support migration, so their migration frequency is zero.

\begin{figure}[t]
	\centering
	\subfloat[][LLaMA-13B on NVIDIA V100]{
		\begin{minipage}[t]{0.45\linewidth}
			\centering
			\includegraphics[width=\linewidth]{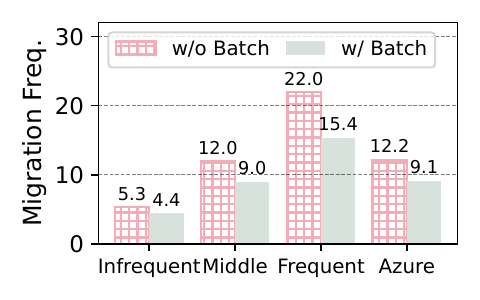}
		\end{minipage}
		\label{fig:A100_Batch}
	}
    \hfill
	\subfloat[][LLaMA-7B on NVIDIA 4090]{
		\begin{minipage}[t]{0.45\linewidth}
			\centering
			\includegraphics[width=\linewidth]{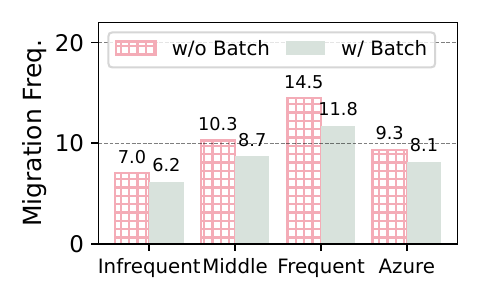}
		\end{minipage}%
		\label{fig:4090_Batch}
	}
	\centering
	\caption{The performance improvement caused by request operation batching in \textsc{Mell} under different workloads.}
    \label{fig:batch}
\end{figure}

\begin{figure}[t]
	\centering
	\subfloat[][LLaMA-13B on NVIDIA A100]{
		\begin{minipage}[t]{\linewidth}
			\centering
			\includegraphics[width=0.95\linewidth]{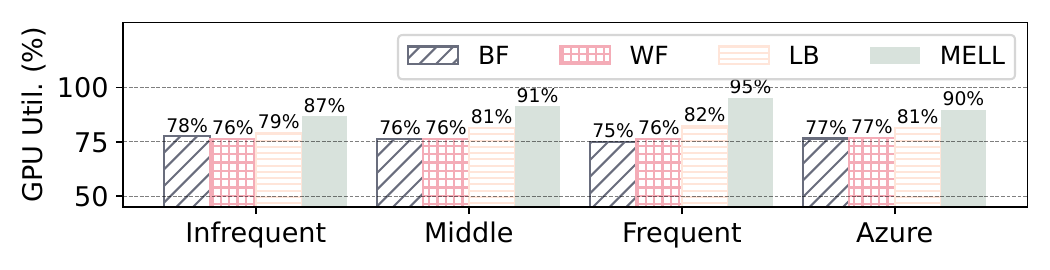}
		\end{minipage}
		\label{fig:A100_13B_move_num}
	}
	\\
	\subfloat[][LLaMA-7B on NVIDIA 4090]{
		\begin{minipage}[t]{\linewidth}
			\centering
			\includegraphics[width=0.95\linewidth]{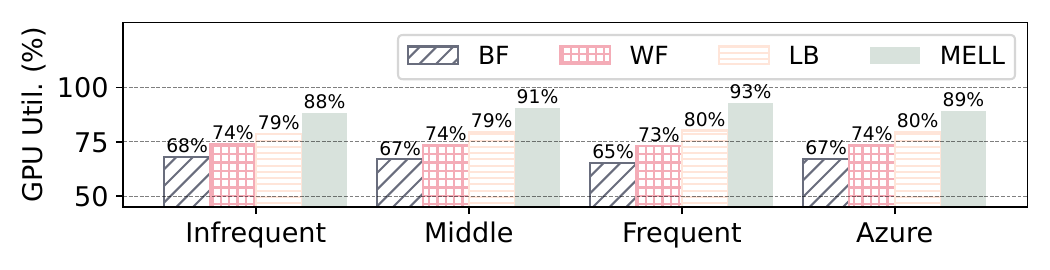} 
		\end{minipage}%
		\label{fig:4090_7B_load}
	}
	\centering
	\caption{The GPU utilization in different systems under the Poisson and Azure workloads.}
    \label{fig:gpu_utilization}
\end{figure}

We evaluate the GPU utilization in various LLM serving systems under different workloads.
As shown in \autoref{fig:gpu_utilization}, \textsc{Mell} consistently achieves the highest GPU utilization across various workloads, with its lowest average at $88\%$. 
In comparison, the peak GPU utilization for other algorithms is around $80\%$, while BF's GPU utilization is as low as $65\%$ in the Poisson workload with high arrival frequency. 
\textsc{Mell} improves the GPU utilization by $8\% \sim 28\%$ compared with the existing systems.
This stark difference underscores the substantial GPU memory waste attributable to fragmented storage spaces. 
\textsc{Mell} addresses this inefficiency by effectively consolidating fragmented spaces through targeted migration requests, optimizing GPU resource utilization.

\autoref{fig:GPU_usage} shows the record of GPU usage of each system under the Poisson workload. 
All algorithms exhibit similar performance in the initial phase. It is evident that there are considerable fluctuations during the service phase, and the fluctuation trends are broadly similar across differing strategies. However, \textsc{Mell} consistently maintains the lowest GPU requirements throughout all these fluctuations.

\begin{figure}[t]
	\centering
	\subfloat[][LLaMA-13B on NVIDIA A100]{
		\begin{minipage}[t]{\linewidth}
			\centering
			\includegraphics[width=0.95\linewidth]{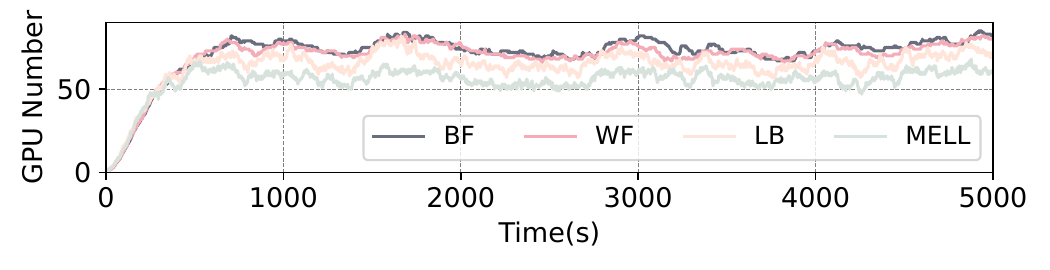}
		\end{minipage}
		\label{fig:A100_13B_gpu_usage}
	}\vspace{-0.2cm}
        \\
	\subfloat[][LLaMA-7B on NVIDIA 4090]{
		\begin{minipage}[t]{\linewidth}
			\centering
			\includegraphics[width=0.95\linewidth]{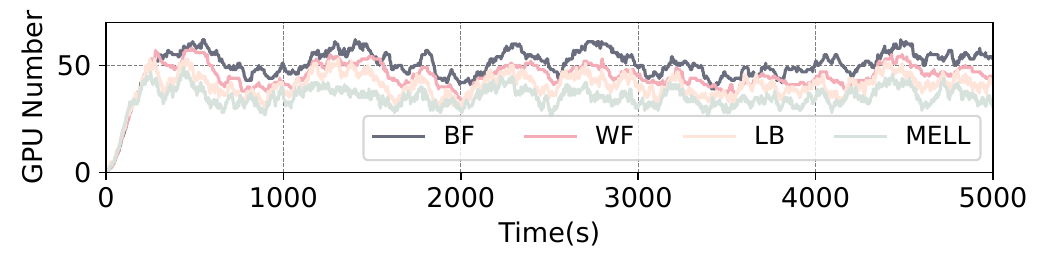} 
		\end{minipage}%
		\label{fig:4090_7B_move_num}
	}
	\centering
	\caption{The number of GPUs at different times in each systems under Poisson workload with high arrival frequency.}
    \label{fig:GPU_usage}
\end{figure}

\section{Conclusion}

This paper proposes \textsc{Mell}, a memory-efficient LLM serving system via multi-GPU KV cache management. 
The system comprises an adaptive request migration mechanism for dynamic resource levels and an online KV cache scheduling algorithm that reduces the number of GPUs with limited request migration.
We implement a prototype of \textsc{Mell} on LLaMA and vLLM and evaluate it based on real chatbot conversations.
The results show that \textsc{Mell} reduces the number of GPUs by $9\% \sim 31\%$ and increases the GPU utilization by $10\% \sim 43\%$ on a Poisson simulated workload and a real workload from Azure compared to the existing LLM serving systems.
In future work, we will investigate multi-GPU KV cache management for LLMs with parameter sizes that exceed the capacity of a single GPU.

\clearpage
\bibliographystyle{IEEEtran}
\bibliography{ref}

\end{document}